\documentclass[a4paper]{easychair}

\usepackage{amssymb}
\usepackage{bcprules}
\usepackage{prooftree}
\usepackage[normalem]{ulem}
\usepackage{amsthm}

\lstdefinelanguage{Scala}%
{morekeywords={abstract,case,catch,char,class,%
    def,else,extends,final,%
    if,import,%
    match,module,new,null,object,override,package,private,protected,%
    public,return,super,this,throw,trait,try,type,val,var,with,implicit,%
    macro,sealed,%
  },%
  sensitive,%
  morecomment=[l]//,%
  morecomment=[s]{/*}{*/},%
  morestring=[b]",%
  morestring=[b]',%
  showstringspaces=false%
}[keywords,comments,strings]%

\newcommand{\ie}{{\em i.e.,~}}

\newcommand{\ok}{{~\textbf{ok}}}

\newcommand{\gap}{\quad\quad}

\newcommand{\ba}{\begin{array}}
\newcommand{\ea}{\end{array}}

\newcommand{\ei}{\end{array}}
\newcommand{\bcases}{\left\{\begin{array}{ll}}
\newcommand{\ecases}{\end{array}\right.}

\newcommand{\typ}{:}
\newcommand{\seq}[1]{\overline{#1}}
\newcommand{\aframe}[3]{\langle #1, #2 \rangle ^{#3}}

\newcommand{\obj}[2]{\langle #1, #2 \rangle}

\newcommand{\reduce}[4]{#1, #2 \;\leadsto\; #3, #4}

\newcommand{\fsreduce}[4]{#1, #2 \;\twoheadrightarrow\; #3, #4}
\newcommand{\fsreducebreak}[4]{#1, #2 \\ \;\twoheadrightarrow\; #3, #4}

\newcommand{\freduce}[4]{#1, #2 \;\longrightarrow\; #3, #4}

\theoremstyle{plain}
\newtheorem{theorem}{Theorem}

\theoremstyle{definition}
\newtheorem{definition}{Definition}

\theoremstyle{remark}

\begin{document}

\title{A Formal Model for Direct-style Asynchronous Observables}
\titlerunning{A Formal Model for Direct-style Asynchronous Observables}

\author{Philipp Haller\inst{1}
\and Heather Miller\inst{2}}

\institute{
  KTH Royal Institute of Technology, Sweden\\
  \email{phaller@kth.se}
\and
  EPFL, Switzerland\\
  \email{heather.miller@epfl.ch}
}

\authorrunning{Haller and Miller}

%%%%%%%%%%%%%%%%%%%%%%%%%%%%%%%%%%%%%%%%%%%%%%%%%%%
\maketitle
%%%%%%%%%%%%%%%%%%%%%%%%%%%%%%%%%%%%%%%%%%%%%%%%%%%

\begin{abstract}

Languages like F\#, C\#, and recently also Scala, provide ``Async''
programming models which aim to make asynchronous programming easier by
avoiding an inversion of control that is inherent in callback-based
programming models. This paper presents a novel approach to integrate the
Async model with observable streams of the Reactive Extensions model. Reactive
Extensions are best-known from the .NET platform, and widely-used
implementations of its programming model exist also for Java, Ruby, and other
languages. This paper contributes a formalization of the unified ``Reactive
Async'' model in the context of an object-based core calculus. Our formal
model captures the essence of the protocol of asynchronous observables using a
heap evolution property. We prove a subject reduction theorem; the theorem
implies that reduction preserves the heap evolution property. Thus, for
well-typed programs our calculus ensures the protocol of asynchronous
observables.

\end{abstract}

\section{Introduction}

Asynchronous programming has been a challenge for a long time. A multitude of
programming models have been proposed that aim to simplify the task.
Interestingly, there are elements of a convergence arising, at least with
respect to the basic building blocks: futures and promises have begun to play
a more and more important role in a number of languages like Java, C++,
ECMAScript, and Scala.

The Async extensions of F\#~\cite{SymePL11}, C\#~\cite{FormalizingAsync}, and
Scala~\cite{ScalaAsyncSIP} provide language support for programming with
futures (or ``tasks''), by avoiding an inversion of control that is inherent
in designs based on callbacks. However, these extensions are so far only
applicable to futures or future-like abstractions. In this paper we present
an integration of the Async model with a richer underlying abstraction, the
observable streams of the Reactive Extensions model.~\cite{RxCACM} A reactive
stream is a stream of {\em observable events} which an arbitrary number of
{\em observers} can subscribe to. The set of possible event patterns of
observable streams is strictly greater than those of futures. A stream can (a)
produce zero or more regular events, (b) complete normally, or (c) complete
with an error (it's even possible for a stream to never complete.) Given the
richer substrate of reactive streams, the Async model has to be generalized in
several dimensions.

This paper makes the following contributions:
\begin{itemize}

  \item A design of a new programming model, RAY, which integrates the Async
        model and the Reactive Extensions model (the name RAY is inspired by the main
        constructs, \texttt{rasync}, \texttt{await} and \texttt{yield});

  \item A formal model of the proposed programming model. Our operational
        semantics extends the formal model presented in~\cite{FormalizingAsync}
        for C\#'s async/await to observable streams. The formal model captures the
        essence of the protocol of asynchronous observables using a heap evolution
        property;

  \item A proof of subject reduction for the presented core calculus. The theorem
        implies that reduction preserves the heap evolution property. Thus, for
        well-typed programs our calculus ensures the protocol of asynchronous
        observables.

\end{itemize}

The rest of the paper is organized as follows. The following
Section~\ref{sec:background} provides background on Scala Async, an implementation of the
Async model, and the Reactive Extensions model. Section~\ref{sec:model}
introduces our unified Reactive Async model. In
Section~\ref{sec:formalization} we present a formalization of the Reactive
Async model in the context of an object-based core calculus. Section~\ref{sec:correctness}
presents correctness properties including subject reduction.

\section{Background}\label{sec:background}

\vspace{0.2cm}
\subsection{Scala Async}\label{sec:scala-async}

Scala Async provides constructs that aim to facilitate programming with
asynchronous events in Scala. The introduced constructs are inspired to a
large extent by extensions that have been introduced in C\# version 5~\cite{Hejlsberg:2011:CPL} in a
similar form. The goal is to enable expressing
asynchronous code in ``direct style'', \ie in a familiar blocking style
where suspending operations look as if they were blocking while at the same
time using efficient non-blocking APIs under the hood.

In Scala, an immediate consequence is that non-blocking code using Scala's
futures API~\cite{ScalaFuturesSIP} does not have to resort to (a) low-level
callbacks, or (b) higher-order functions like \verb|map| and \verb|flatMap|. While the
latter have great composability properties, they can appear unnatural when
used to express the regular control flow of a program.

For example, an efficient non-blocking composition of asynchronous web service
calls using futures can be expressed as follows in Scala:

\lstset{numbers=left,xleftmargin=2em}
\begin{lstlisting}
val futureDOY: Future[Response] =
  WS.url("http://api.day-of-year/today").get

val futureDaysLeft: Future[Response] =
  WS.url("http://api.days-left/today").get

futureDOY.flatMap { doyResponse =>
  val dayOfYear = doyResponse.body
  futureDaysLeft.map { daysLeftResponse =>
    val daysLeft = daysLeftResponse.body
    Ok("" + dayOfYear + ": " +
            daysLeft + " days left!")
  }
}
\end{lstlisting}

Line 1 and 4 define two futures obtained as results of asynchronous requests
to two hypothetical web services using an API inspired by the Play! Framework (for
the purpose of this example, the definition of type \verb|Response| is
unimportant).

This can be expressed more intuitively in direct style using Scala Async as
follows (this example is adopted from the SIP proposal~\cite{ScalaAsyncSIP}):

\begin{lstlisting}
val respFut = async {
  val dayOfYear = await(futureDOY).body
  val daysLeft = await(futureDaysLeft).body
  Ok("" + dayOfYear + ": " +
          daysLeft + " days left!")
}
\end{lstlisting}

The \verb|await| on line 2 causes the execution of the \verb|async|
block to suspend until \verb|futureDOY| is completed (with a successful result
or with an exception). When the future is completed successfully, its result
is bound to the \verb|dayOfYear| local variable, and the execution of the
\verb|async| block is resumed. When the future is completed with an exception
(for example, because of a timeout), the invocation of \verb|await| re-throws
the exception that the future was completed with. In turn, this completes
future \verb|respFut| with the same exception. Likewise, the \verb|await| on
line 3 suspends the execution of the \verb|async| block until
\verb|futureDaysLeft| is completed.

The main methods provided by Scala Async, \verb|async| and \verb|await|, have
the following type signatures:

\lstset{numbers=none,xleftmargin=0em}
\begin{lstlisting}
def async[T](body: => T): Future[T]
def await[T](future: Future[T]): T
\end{lstlisting}

Given the above definitions, \verb|async| and \verb|await| ``cancel each other
out:''

\begin{lstlisting}
await(async { <expr> }) = <expr>
\end{lstlisting}

This ``equation'' paints a grossly over-simplified picture, though, since the
actual operational behavior is much more complicated: \verb|async| typically
schedules its argument expression to run asynchronously on a thread pool;
moreover, \verb|await| may only be invoked within a syntactically enclosing
\verb|async| block.

\subsection{Reactive Extensions}

The Rx programming model is based on two interface traits: \verb|Observable|
and \verb|Observer|. \verb|Observable| represents observable streams, \ie
streams that produce a sequence of events. These events can be observed by
registering an \verb|Observer| with the \verb|Observable|. The \verb|Observer|
provides methods which are invoked for each of the kinds of events produced by
the \verb|Observable|. In Scala, the two traits can be defined as shown in
Figure~\ref{fig:observable-observer}.

\begin{figure}[ht!]
  \centering
  \lstset{numbers=none,xleftmargin=0em}
  \begin{lstlisting}
  trait Observable[T] {
    def subscribe(obs: Observer[T]): Closable
  }

  trait Observer[T] extends (Try[T] => Unit) {
    def apply(tr: Try[T]): Unit
    def onNext(v: T) = apply(Success(v))
    def onFailure(t: Throwable) = apply(Failure(t))
    def onDone(): Unit
  }
  \end{lstlisting}
  \caption{The \texttt{Observable} and \texttt{Observer} traits.}
  \label{fig:observable-observer}
\end{figure}

The idea of the \verb|Observer| is that it can respond to three different
kinds of events, (1) the next regular event (\verb|onNext|), (2) a failure
(\verb|onFailure|), and (3) the end of the observable stream (\verb|onDone|).
Thus, the two traits constitute a variation of the classic subject/observer
pattern~\cite{EugsterFGK03}. Note that \verb|Observable|'s \verb|subscribe|
method returns a \verb|Closable|; it has only a single abstract \verb|close|
method which removes the subscription from the observable. The next listing
shows an example implementation.

Note that in our Scala version the \verb|Observer| trait extends the function
type \verb|Try[T] => Unit|. \verb|Try[T]| is a simple container type which
supports heap-based exception handling (as opposed to the traditional stack-based
exception handling using expressions like \verb|try-catch-finally|.)
There are two subclasses of \verb|Try[T]|: \verb|Success| (encapsulating a
value of type \verb|T|) and \verb|Failure| (encapsulating an exception). Given
the above definition, a concrete \verb|Observer| only has to provide
implementations for the \verb|apply| and \verb|onDone| methods. Since
\verb|apply| takes a parameter of type \verb|Try[T]| its implementation
handles the \verb|onNext| and \verb|onFailure| events all at once (in Scala,
this is tyically done by pattern matching on \verb|tr| with cases for
\verb|Success| and \verb|Failure|).

The \verb|Observer| and \verb|Observable| traits are used as follows. For
example, here is a factory method for creating an observable from a text input
field of typical GUI toolkits (this example is adapted from~\cite{RxCACM}):

\lstset{numbers=none,xleftmargin=0em}
\begin{lstlisting}
def textChanges(tf: JTextField): Observable[String] =
  new ObservableBase[String] {
    def subscribe(o: Observer[String]) = {
      val l = new DocumentListener {
        def changedUpdate(e: DocumentEvent) = {
          o.onNext(tf.getText())
        }
      }
      tf.addDocumentListener(l)
      new Closable() {
        def close() = {
          tf.removeDocumentListener(l)
        }
      }
    }
  }
\end{lstlisting}

This newly-defined \verb|textChanges| combinator can be used with other Rx
combinators as follows:

\begin{lstlisting}
textChanges(input)
.flatMap(word => completions(word))
.subscribe(observeChanges(output))
\end{lstlisting}

We start with the observable created using the \verb|textChanges| method from
above. Then we use the \verb|flatMap| combinator (called \verb|Select| in C\#)
to transform the observable into a new observable which is a stream of
completions for a given word (a string). On the resulting observable we call
\verb|subscribe| to register a consumer: \verb|observeChanges| creates an
observer which outputs all received events to the \verb|output| stream. (The
shown example suffers from a problem explained in~\cite{RxCACM} which
motivates the use of an additional \verb|Switch| combinator which is omitted
here for brevity.)

\section{The Reactive Async Model}\label{sec:model}

This Section provides an (example-driven) overview of the Reactive Async Model
which integrates the Async Model and the Reactive Extensions Model.

The basic idea is to generalize the Async model, so that it can be used not
only with futures, but also with observable streams. This means, we need
constructs that can create observables, as opposed to only futures (like
\verb|async|), and we need ways to wait for more events than just the
completion of a future. Essentially, it should be possible to await all kinds
of events produced by an observable stream. Analogous to \verb|await| which
waits for the completion event of a future, we introduce variations like
\verb|awaitNext| and \verb|awaitNextOrDone| to express waiting for the events
of an observable stream.

\subsection{Example}

The following example shows how to await a fixed number of events of a stream
in the Reactive Async Model:

\begin{lstlisting}
val obs = rasync {
  var events = List[Int]()
  while (events.size < 5) {
    val event = awaitNext(stream)
    events = event :: events
  }
  Some(events)
}
\end{lstlisting}

Note that we are using the \verb|rasync| construct; it is a generalized
version of the \verb|async| construct of Section~\ref{sec:scala-async} which
additionally supports methods to await events of observable streams.

In the above example, the invocation of \verb|awaitNext| suspends the
\verb|rasync| block until the producer of \verb|stream| calls \verb|onNext| on
its observers. The argument of this \verb|onNext| call (the next event) is
returned as a result from \verb|awaitNext|. The result of \verb|rasync|,
\verb|obs|, has type \verb|Observable[List[Int]]|. Once the body of
\verb|rasync| has been fully evaluated, \verb|obs| publishes two events:
first, an \verb|onNext| event which carries \verb|events| (the list with five
elements), and second, an \verb|onDone| event; it is not possible for
\verb|obs| to publish further events.

Note that the result of an \verb|rasync| block has a type of the form
\verb|Option[T]|; in the case where this optional value is empty
(\verb|None|), only an \verb|onDone| event is published as a result of fully
evaluating the \verb|rasync| block. (It is, however, possible to publish other
events beforehand, as shown in the following sections.) Otherwise, the
semantics of \verb|rasync| is analogous to the behavior of a regular
\verb|async| block: when its body has been fully evaluated, the future, which
is the result of \verb|async|, is completed and further changes to the state
of the future are impossible.

\subsection{Awaiting Stream Termination}

Sometimes it is not known statically how many events a stream might still
publish. One might want to collect all events until the stream is done
(finished publishing events). In this case it is necessary to have a way to
wait for either of two events: the stream publishes a next event, or the
stream is done. This can be supported using a method \verb|awaitNextOrDone|
which returns an \verb|Option[T]| when applied to an \verb|Observable[T]|:

\begin{lstlisting}
rasync {
  var events: List[Int] = List()
  var next: Option[Int] = awaitNextOrDone(stream)
  while (next.nonEmpty) {
    events = next.get :: events
    next = awaitNextOrDone(stream)
  }
  Some(events)
}
\end{lstlisting}

In the above example, the body of \verb|rasync| repeatedly waits for the given
\verb|stream| to publish either a next event or to reach its end, using
\verb|awaitNextOrDone|. As long as the \verb|stream| continues to publish
events (in which case \verb|next| of type \verb|Option[Int]| is non-empty),
each event is prepended to the \verb|events| list; this list is the single
event that the observable which is, in turn, created by \verb|rasync|
publishes (once the body of \verb|rasync| has been fully evaluated).

\subsection{Creating Complex Streams}

The streams created by \verb|rasync| in the previous sections are rather
simple: after consuming events from other streams only a single interesting
event is published on the created stream (by virtue of reaching the end of the
\verb|rasync| block). In this section, we explain how more complex streams can
be created in the Reactive Async Model.

\begin{figure}[ht!]
  \centering
  \begin{lstlisting}
  val forwarder = rasync[Int] {
    var next: Option[Int] =
      awaitNextOrDone(stream)

    while (next.nonEmpty) {
      yieldNext(next.get)
      next = awaitNextOrDone(stream)
    }
    None
  }
  \end{lstlisting}
  \caption{A simple forwarder stream.}
  \label{fig:forwarder}
\end{figure}

Suppose we would like to create a stream which simply publishes an event for
each event observed on another \verb|stream|. In this case, the constructs we
have seen so far are not sufficient, since an arbitrary number of events have
to be published from within the \verb|rasync| block. This is where the new
method \verb|yieldNext| comes in: it publishes the next event to the stream
returned by \verb|rasync|. Our simple forwarder example can then be expressed
as shown in Figure~\ref{fig:forwarder}.

Note that in the above example, the result of the body of the \verb|rasync|
block is \verb|None|; consequently, the resulting \verb|forwarder| stream only
publishes an \verb|onDone| event when \verb|rasync|'s body has been fully
evaluated. In this case, it is assumed that the only ``interesting'' non-done
events of \verb|forwarder| are published using \verb|yieldNext|.

\begin{figure}
  \centering
$\ba[t]{l@{\hspace{2mm}}l}
p    ::=  \seq{cd}~e                                       & \mbox{program}             \\
cd   ::=  \texttt{class}~C~\{\seq{fd}~\seq{md}\}           & \mbox{class declaration}   \\
fd   ::=  \texttt{var}~f: \sigma                           & \mbox{field declaration}   \\
md   ::=  \texttt{def}~m(\seq{x: \sigma}): \tau = e        & \mbox{method declaration}  \\
\sigma,\tau ::=                                            & \mbox{type}                \\
\gap ~|~  \gamma                                           & \gap\mbox{value type}      \\
\gap ~|~  \rho                                             & \gap\mbox{reference type}  \\
\gamma ::=                                                 & \mbox{value type} \\
\gap ~|~  \texttt{Boolean}                                 & \gap\mbox{boolean}         \\
\gap ~|~  \texttt{Int}                                     & \gap\mbox{integer}         \\
\rho ::=                                                   & \mbox{reference type}      \\
\gap ~|~  C                                                & \gap\mbox{class type}      \\
\gap ~|~  \texttt{Observable$[\sigma]$}                    & \gap\mbox{observable type} \\
\ea$
  \caption{Core language syntax. $C$ is a class name, $f,m$ are field and
    method names.}
  \label{fig:lang-syntax}
\end{figure}

\section{Formalization}\label{sec:formalization}

One of the contributions of this paper is an operational semantics of the
proposed programming model. Our operational semantics generalizes the formal
model presented in~\cite{FormalizingAsync}.

\subsection{Syntax}

Figure~\ref{fig:lang-syntax} and Figure~\ref{fig:lang-syntax-2} show the
syntax of our core language. Note that programs are written in {\em A-normal
form} (ANF) which forces all subexpressions to be named; this simplifies the
presentation of the operational semantics. Note that our core language does
not support any form of subtyping; thus, class declarations do not specify a
superclass. This is adopted from~\cite{FormalizingAsync}; the presented
reactive features are orthogonal to subtyping.

\begin{figure}
  \centering
$\ba[t]{l@{\hspace{2mm}}l}
e    ::=                                                 & \mbox{expressions}      \\
\gap ~|~ \uline{b}                                       & \gap\mbox{boolean}      \\
\gap ~|~ \uline{i}                                       & \gap\mbox{integer}      \\
\gap ~|~ x                                               & \gap\mbox{variable}     \\
\gap ~|~ \texttt{null}                                   & \gap\mbox{null}         \\
\gap ~|~ \texttt{if}~(x)~\{e\}~\texttt{else}~\{e'\}      & \gap\mbox{condition}    \\
\gap ~|~ \texttt{while}~(x)~\{e\}                        & \gap\mbox{while loop}   \\
\gap ~|~ x.f                                             & \gap\mbox{selection}    \\
\gap ~|~ x.f = y                                         & \gap\mbox{assignment}   \\
\gap ~|~ x.m(\seq{y})                                    & \gap\mbox{invocation}   \\
\gap ~|~ \texttt{new C}(\seq{y})                         & \gap\mbox{instance creation}   \\
\gap ~|~ \texttt{let}~x = e~\texttt{in}~e'               & \gap\mbox{let binding}  \\
\gap ~|~ \texttt{rasync[}\sigma\texttt{]}(\bar{y})~\{e\} & \gap\mbox{observable creation} \\
\gap ~|~ \texttt{await}(x)                               & \gap\mbox{await next event}    \\
\gap ~|~ \texttt{yield}(x)                               & \gap\mbox{yield event}  \\
\ea$
  \caption{RAY expressions.}
  \label{fig:lang-syntax-2}
\end{figure}

A RAY program consists of a collection of class definitions, as well as
the body (an expression) of a ``main'' method. A class \texttt{C} has (a
possibly empty) sequence of  public fields and methods, $\seq{f}$ with types
$\seq{\sigma}$, and $\seq{md}$, respectively. Methods are
public with return type $\tau$, a type which may represent either a value or
reference type, and their body is an expression $e$.

The \texttt{Observable[$\sigma$]} family of types is used to model the
generic nature of observables. They represent observables such that an
\texttt{rasync} block can await their next event using \texttt{await}.
Conversely, inside the body $e$ of an expression
$\texttt{rasync[}\sigma\texttt{]}(\bar{y})~\{e\}$, \texttt{yield} can be
used to publish events of type $\sigma$.

% TODO summarize expressions in Figure~\ref{fig:lang-syntax-2}

% Expressions in RAY include constants, which can be either an integer $i$,
% boolean $b$, or the \texttt{null} literal. They may also be represented by
% class declarations, selections, or invocations. Here, $x$ and $y$
% represent variable names, while $f$ ranges over field names and $m$ ranges
% over method names.

\subsection{Operational Semantics}

\subsubsection{Notation}

A heap, denoted $H$, partially maps object identifiers (ranged over by $o$) to
heap objects, denoted $\langle\texttt{C}, FM\rangle$, representing a pair of
type $\texttt{C}$ and a field map, $FM$. A field map partially maps fields $f$
to values (ranged over by $v$), where $v$ can be either an integer, a boolean,
\texttt{null}, or an object identifier (the address of an object in the heap).

Frames have the form $\aframe L e l$ where $L$ maps local variables to
their values, $e$ is an expression, and $l$ is a label. A
label is either $s$ denoting a regular, synchronous frame, or $a(o, \bar{p})$ denoting
an asynchronous frame; in this case, $o$ is the heap address of a
corresponding observable object $\langle \texttt{Observable$[\sigma]$},
running(\bar{F}, \bar{S}) \rangle$; $\bar{p}$ is a sequence of object identifiers of
observables that observable $o$ has subscribed to. $\bar{F}$ is a set of
asynchronous frames, namely, all observables that are currently suspended awaiting $o$
to publish a new event. $\bar{S}$ is a set of subscribers which are explained
below.

There are three kinds of transition rules. The first kind goes from a heap and
a frame to a new heap and a new frame (simple right arrow). The second kind
goes from a heap and a frame stack to a new heap and a new frame stack (double
right arrow). The third kind goes from a heap and a set of frame stacks to a
new heap and a new set of frame stacks (squiggly right arrow).

\subsubsection{Synchronous Transition Rules}

\begin{figure*}[ht!]
  \centering

\infrule[\textsc{E-Var}]
{ 
}
{ \freduce H {\aframe L {\texttt{let}~x = y~\texttt{in}~e} l} H {\aframe {L[x \mapsto L(y)]} e l}
}

\vspace{0.3cm}

\infrule[\textsc{E-Field}]
{  H(L(y)) = \obj {\rho} {FM}
}
{ \freduce H {\aframe L {\texttt{let}~x = y.f~\texttt{in}~e} l} H {\aframe {L[x \mapsto FM(f)]} e l}
}

\vspace{0.3cm}

\infrule[\textsc{E-While}]
{  F' =
    \begin{cases}
      \aframe L {\texttt{let}~x' = e~\texttt{in}~\texttt{let}~x = \texttt{while}~(y)~\{e\}~\texttt{in}~e'} l \quad \text{if}~L(y) = \texttt{true} \andalso x' \notin dom(L) \\
      \aframe L {\texttt{let}~x = \texttt{false}~\texttt{in}~e'} l \quad \text{if}~L(y) = \texttt{false} \\
    \end{cases} \\
}
{ \freduce H {\aframe L {\texttt{let}~x = \texttt{while}~(y)~\{e\}~\texttt{in}~e'} l} H {F'}
}

\vspace{0.3cm}

\infrule[\textsc{E-Cond}]
{  F' =
    \begin{cases}
      \aframe L {\texttt{let}~x = t~\texttt{in}~u} l \quad \text{if}~L(y) = \texttt{true} \\
      \aframe L {\texttt{let}~x = s~\texttt{in}~u} l \quad \text{if}~L(y) = \texttt{false} \\
    \end{cases} \\
}
{ \freduce H {\aframe L {\texttt{let}~x = \texttt{if}~(y)~\{t\}~\texttt{else}~\{s\}~\texttt{in}~u} l} H {F'}
}

\vspace{0.3cm}

\infrule[\textsc{E-Assign}]
{ L(x) = o \andalso H(o) = \obj {\sigma} {FM} \andalso H' = H[o \mapsto \obj {\sigma} {FM[f \mapsto L(y)]}] \\
}
{ \freduce H {\aframe L {\texttt{let}~x' = x.f = y~\texttt{in}~e} l} {H'} {\aframe L {\texttt{let}~x' = y~\texttt{in}~e} l}
}

\vspace{0.3cm}

\infrule[\textsc{E-New}]
{ fields(C) = \bar{f} \andalso o \notin dom(H) \\
  H' = H[o \mapsto \obj C {\bar{f} \mapsto L(\bar{y})}] \\
}
{ \freduce H {\aframe L {\texttt{let}~x =~\texttt{new}~C(\bar{y})~\texttt{in}~e} l} {H'} {\aframe {L[x \mapsto o]} e l}
}

  \caption{Simple frame transition rules.}
  \label{fig:opsem-rules-simple}
\end{figure*}

Figure~\ref{fig:opsem-rules-simple} shows simple frame transition rules. Note
that all transition rules preserve the labels of frames. Rule \textsc{E-Var}
looks up the value $L(y)$ of $y$ in local variable mapping $L$; the local
variable mapping of the target frame maps $x$ to $L(y)$; reduction continues
with expression $e$. Rule \textsc{E-Field} looks up the value of field $y.f$
using $L$ and $H$; as before, reduction continues with $e$. Rule
\textsc{E-While} looks up the value of the condition variable in the local
variable mapping and continues reduction accordingly. Note that our core
language does not distinguish between expressions and statements; thus, also
the while loop is an expression which reduces to the constant \verb|false|
when the loop condition is \verb|false|. Rule \textsc{E-Cond} is standard.
Rule \textsc{E-Assign} combines the heap $H$, local variable mapping $L$, and
field mapping $FM$ in the natural way for field assignment. Rule
\textsc{E-New} creates a new instance of class $C$, assigning the constructor
arguments $\bar{y}$ to the fields of the new instance.

\begin{figure*}[ht!]
  \centering

\infrule[\textsc{E-Method}]
{ H(L(y)) = \obj {\rho} {FM} \andalso mbody(\rho, m) = (\seq{x}) \rightarrow e' \\
  L' = [\bar{x} \mapsto L(\bar{z}), \texttt{this} \mapsto L(y)]
}
{ \fsreduce H {\aframe L {\texttt{let}~x = y.m(\bar{z})~\texttt{in}~e} l \circ FS} H {\aframe {L'} {e'} s \circ {\aframe L e l}_x \circ FS}
}

\vspace{0.3cm}

\infrule[\textsc{E-Return}]
{
}
{ \fsreduce H {\aframe L y s \circ {\aframe {L'} e l}_x \circ FS} H {\aframe {L'[x \mapsto L(y)]} e l \circ FS}
}

\vspace{0.3cm}

\infrule[\textsc{E-Frame}]
{ \freduce H F {H'} {F'}
}
{ \fsreduce H {F \circ FS} {H'} {F' \circ FS}
}

  \caption{Method call and return transition rules.}
  \label{fig:opsem-rules-sync}
\end{figure*}

Figure~\ref{fig:opsem-rules-sync} shows the transition rules for method call
and return. Rule \textsc{E-Method} evaluates a method invocation. The run-time
type of the receiver, $\rho$, is looked up in heap $H$. Using the auxiliary
function $mbody$ we look up the body of method $m$ in $\rho$. To evaluate the
method body, a new frame with synchronous label $s$ is created and pushed on
top of the frame stack. Importantly, the caller frame (with expression $e$) is
annotated with variable $x$; this annotation is used for the transfer of the
return value as follows. Rule \textsc{E-Return} shows how a value is returned
from a method invocation to the caller. A method call returns when the
expression of its frame has been reduced to a variable $y$. The method's frame
is popped off the frame stack, and the frame of the caller is replaced with a
frame that maps variable $x$ to the value of $y$. Crucially, the frame of the
caller is annotated with $x$. Finally, rule \textsc{E-Frame} transitions a
{\em frame stack} $F \circ FS$ by transitioning frame $F$.

\subsubsection{Asynchronous Transition Rules}

\begin{figure*}[ht!]
  \centering

\infrule[\textsc{E-RAsync}]
{  L(\bar{y}) = \bar{p} \andalso o \notin dom(H) \\
   H_0 = H[o \mapsto \langle \texttt{Observable$[\sigma]$}, running(\epsilon, \epsilon) \rangle] \\
   \forall p_i \in \bar{p}.~H(p_i) = \langle \psi_i, running(\bar{F}_i, \bar{S}_i) \rangle \\
   \forall i \in 1 \ldots n.~H_i = H_{i-1}[p_i \mapsto \langle \psi_i, running(\bar{F}_i, \langle o, [] \rangle :: \bar{S}_i) \rangle] \\
   H' = H_n
}
{ \fsreducebreak H {\aframe L {\texttt{let}~x = \texttt{rasync[}\sigma\texttt{]}(\bar{y})~\{e\}~\texttt{in}~t} l \circ FS}
    {H'} {\aframe L e {a(o, \bar{p})} \circ \aframe {L[x \mapsto o]} t l \circ FS}
}

\vspace{0.3cm}

\infrule[\textsc{E-Yield}]
{  H(o) = \langle \texttt{Observable$[\sigma]$}, running(\bar{F}, \bar{S}) \rangle \\
   \bar{R} = resume(\bar{F}, \texttt{Some}(L(z))) \andalso Q = \{ R \circ \epsilon ~|~ R \in \bar{R} \} \\
   \bar{S'} = \{ \langle {o'}, L(z) :: q \rangle ~|~ \langle {o'}, q \rangle \in \bar{S} \} \\
   H' = H[o \mapsto \langle \texttt{Observable$[\sigma]$}, running(\epsilon, \bar{S'}) \rangle]
}
{ \reduce H {\{ \aframe L {\texttt{yield}(z)} {a(o, \bar{p})} \circ FS\} \cup P}
    {H'} {\{ \aframe L {z} {a(o, \bar{p})} \circ FS \} \cup P \cup Q}
}

\vspace{0.3cm}

\infrule[\textsc{E-RAsync-Return}]
{  H(o) = \langle \texttt{Observable$[\sigma]$}, running(\bar{F}, \bar{S}) \rangle \\
   \bar{R} = resume(\bar{F}, \texttt{None}) \andalso Q = \{ R \circ \epsilon ~|~ R \in \bar{R} \} \\
   H_0 = H[o \mapsto \langle \texttt{Observable$[\sigma]$}, done(\bar{S}) \rangle] \\
   \forall i \in 1 \ldots n.~H_i = H_{i-1}[p_i \mapsto unsub(o, p_i, H)]
}
{ \reduce H {\{ \aframe L {x} {a(o, \bar{p})} \circ FS \} \cup P}
    {H_n} {\{ FS \} \cup P \cup Q}
}

\vspace{0.3cm}

\infrule[\textsc{E-Await1}]
{  F = \aframe L {\texttt{let}~x =~\texttt{await}(y)~\texttt{in}~t} {a(o, \bar{p})} \andalso L(y) = o' \\
   H(o') = \langle \texttt{Observable$[\sigma]$}, running(\bar{F}, \bar{S}) \rangle \andalso
   \bar{S} = \bar{R} \uplus \{ \langle o, [] \rangle \} \\
   H' = H[o' \mapsto \langle \texttt{Observable$[\sigma]$}, running(F :: \bar{F}, \bar{R}) \rangle]
}
{ \fsreduce H {F \circ FS} {H'} {FS}
}

\vspace{0.3cm}

\infrule[\textsc{E-Await2}]
{  L(y) = o' \andalso H(o') = \langle \texttt{Observable$[\sigma]$}, running(\bar{F}, \bar{S}) \rangle \andalso
   \bar{S} = \bar{R} \uplus \{ \langle o, q :: v \rangle \} \\
   H' = H[o' \mapsto \langle \texttt{Observable$[\sigma]$}, running(\bar{F}, \bar{R} \cup \{ \langle o, q \rangle \}) \rangle]
}
{ \fsreduce H {\aframe L {\texttt{let}~x =~\texttt{await}(y)~\texttt{in}~t} {a(o, \bar{p})} \circ FS}
    {H'} {\aframe {L[x \mapsto \texttt{Some}(v)]} t {a(o, \bar{p})} \circ FS}
}

\vspace{0.3cm}

\infrule[\textsc{E-Await3}]
{  L(y) = o' \andalso H(o') = \langle \texttt{Observable$[\sigma]$}, done(\bar{S}) \rangle \andalso
   \bar{S} = \bar{R} \uplus \{ \langle o, q :: v \rangle \} \\
   H' = H[o' \mapsto \langle \texttt{Observable$[\sigma]$}, done(\bar{R} \cup \{ \langle o, q \rangle \}) \rangle]
}
{ \fsreduce H {\aframe L {\texttt{let}~x =~\texttt{await}(y)~\texttt{in}~t} {a(o, \bar{p})} \circ FS}
    {H'} {\aframe {L[x \mapsto \texttt{Some}(v)]} t {a(o, \bar{p})} \circ FS}
}

  \caption{Asynchronous transition rules.}
  \label{fig:opsem-rules-async}
\end{figure*}

Figure~\ref{fig:opsem-rules-async} shows the asynchronous transition rules.
These rules transition either between frame stacks ($\twoheadrightarrow$) or
processes ($\leadsto$).

Rule \textsc{E-RAsync} evaluates the creation of a new observable using an
expression of the form $\texttt{rasync[}\sigma\texttt{]}(\bar{y})~\{e\}$. The
type argument $\sigma$ determines the type of the observable,
$\texttt{Observable$[\sigma]$}$. The value arguments $\bar{y}$ refer to
observables at addresses $\bar{p}$ in heap $H$.

Recall the representation of observables in the heap. A newly-created
observable has the form $\langle \texttt{Observable$[\sigma]$},
running(\epsilon, \epsilon) \rangle$. In general, a running observable (\ie
an observable that has not terminated, yet) has the form $\langle \psi,
running(\bar{F}, \bar{S}) \rangle$. Each {\em waiter} $F \in \bar{F}$ is an
asynchronous frame of an observable waiting to receive an event. Each {\em
subscriber} $S \in \bar{S}$ is a pair $\langle o, q \rangle$; $o$ is the
address of a subscribed observable that is running, but currently not waiting
to receive an event; $q$ is a queue of events received by $o$ (but not yet
processed).

Rule \textsc{E-RAsync} adds the newly-created observable $o$ as a subscriber
(with an empty queue) to each observable $p_i \in \bar{p}$. Finally, a new
asynchronous frame $\aframe L e {a(o, \bar{p})}$ is pushed on to the frame
stack.

\begin{figure*}
  \centering

\begin{equation*}
resume(\bar{F}, v) = \{ \aframe {L[x \mapsto v]} t {a(o, \bar{p})} ~|~ \aframe L {\texttt{let}~x = \texttt{await}(y)~\texttt{in}~t} {a(o, \bar{p})} \in \bar{F} \}
\end{equation*}

\begin{equation*}
unsub(\bar{S}, o) = \{ \langle o' , q \rangle ~|~ \langle o' , q \rangle \in \bar{S} \land o' \neq o \}
\end{equation*}

  \caption{Auxiliary functions.}
  \label{fig:aux-funs}
\end{figure*}

Rule \textsc{E-Yield} implements the built-in \verb|yield| expression which
publishes a new event to all waiters and subscribers of observable $o$.
Waiters $F \in \bar{F}$ are resumed by creating a set of new frame stacks $Q$
based on frames created using the $resume$ auxiliary function (defined in
Figure~\ref{fig:aux-funs}). Rule \textsc{E-Yield} resumes waiters with value
$L(z)$. The state of observable $o$ is updated such that the set of waiters is
empty. The \verb|yield| expression itself is reduced to $z$.

Rule \textsc{E-RAsync-Return} transitions a process where an asynchronous
frame $\aframe L {x} {a(o, \bar{p})}$ is on top of one of the frame stacks.
Since the frame's expression has been reduced to just a variable $x$, this
means the corresponding observable $o$ is done publishing events at this
point. Therefore, the waiters $\bar{F}$ are resumed with result \verb|None|.
The state of observable $o$ is updated to $\langle
\texttt{Observable$[\sigma]$}, done(\bar{S}) \rangle$; value $done(\bar{S})$
indicates that $o$ has transitioned to the terminated state; events published
to subscribers ($\bar{S}$) remain available for consumption, though (see rules for
\verb|await| below). Finally, $o$ unsubscribes from all observables $\bar{p}$.
$unsub(o, p, H)$ is defined as follows; the function makes use of
$unsub(\bar{S}, o)$ which is defined in Figure~\ref{fig:aux-funs}.

\begin{equation*}
unsub(o, p, H) =
  \begin{cases}
    \langle \psi , running(\bar{F}, unsub(\bar{S}, o)) \rangle & \text{if } H(p) = \langle \psi , running(\bar{F}, \bar{S}) \rangle \\
    \langle \psi , done(unsub(\bar{S}, o)) \rangle & \text{if } H(p) = \langle \psi , done(\bar{S}) \rangle \\
  \end{cases} \\
\end{equation*}

Rules \textsc{E-Await1}, \textsc{E-Await2}, and \textsc{E-Await3} implement
the built-in \verb|await| expression. Rule \textsc{E-Await1} adds the
asynchronous frame $F$ of observable $o$ to the waiters of observable $o'$ in
the case where there is no event from $o'$ ready to be consumed by $o$. Rule
\textsc{E-Await2} handles the dual case where observable $o$ immediately
receives an event from  observable $o'$; the subscribers of $o'$ are updated
accordingly in the target heap $H'$. Rule \textsc{E-Await3} handles the case
where observable $o'$ is in a terminated state $done(\bar{S})$. Importantly, a
subscriber queue in $\bar{S}$ may contain an event to be consumed by the
\verb|await|-invoking observable $o$. In case the corresponding subscriber
queue is empty, observable $o$ suspends analogous to rule \textsc{E-Await1}.

% TODO write down rule E-Await4 which handles the last case

\begin{figure*}[ht!]
  \centering

\infrule[\textsc{E-Exit}]
{
}
{ \reduce H {\{ \epsilon \} \cup P} H P
}

\vspace{0.3cm}

\infrule[\textsc{E-Schedule}]
{ \fsreduce H {FS} {H'} {FS'}
}
{ \reduce H {\{ FS \} \cup P} {H'} {\{ FS' \} \cup P}
}

  \caption{Process transition rules.}
  \label{fig:opsem-rules-process}
\end{figure*}

Process transition rules enable reducing frame stacks, \ie threads;
Figure~\ref{fig:opsem-rules-process} shows the transition rules. We use an
interleaved semantics. Rule \textsc{E-Schedule} non-deterministically selects
and transitions a thread; note that the transition may have side effects on
the heap. Rule \textsc{E-Exit} removes threads with empty frame stacks from
the soup of threads.

\section{Correctness Properties}\label{sec:correctness}

We show that well-typed programs satisfy desirable
properties:
\begin{enumerate}
\item {\em Observable protocol}. For example, a terminated observable never
publishes events again; this protocol property is captured by a {\em heap evolution} invariant.
\item {\em Subject reduction}. Reduction of well-typed programs preserves types.
\end{enumerate}

The proofs of these properties are based on a typing relation, as well as
invariants preserved by reduction. To establish the correctness properties we
have to consider non-interference properties for processes, frame stacks,
frames, and heaps; these properties are shown in
Figure~\ref{fig:non-interference}. The application $obsIds(FS)$ of the $obsIds$
auxiliary function returns the set of all object addresses $o$ in labels
$a(o, \bar{p})$ of the asynchronous frames $FS$ (similarly for a single frame $F$).
The $waiters$ function returns the observable ids of the waiting frames of the
running state of a given observable heap object. For an observable heap object
$H(o) = \langle \texttt{Observable$[\sigma]$}, running(\bar{F}, \bar{S}) \rangle$,
$waiters(H(o)) = obsIds(\bar{F})$. To test whether an observable is currently running
(as opposed to done) we use a simple predicate, $Running$. Finally, to express
disjointness of (sets of) heap addresses we use the symbol $\#$.

% Analogously, for the same observable heap object, $subscribers(H(o)) = \bar{S}$.

\begin{figure*}[ht]
  \centering

\infrule[\textsc{EmpFS-ok}]
{
}
{ H \vdash \epsilon~\textbf{ok}
}

\vspace{0.3cm}

\infrule[\textsc{FS-ok}]
{ H \vdash F~\textbf{ok} \quad H \vdash FS~\textbf{ok} \quad obsIds(F)\#obsIds(FS)
}
{ H \vdash F \circ FS~\textbf{ok}
}

\vspace{0.3cm}

\infrule[\textsc{SF-ok}]
{
}
{ H \vdash F^{s}~\textbf{ok}
}

\vspace{0.3cm}

\infrule[\textsc{CSF-ok}]
{ H \vdash F^{s}~\textbf{ok}
}
{ H \vdash F^{s}_{x}~\textbf{ok}
}

\vspace{0.3cm}

\infrule[\textsc{AF-ok}]
{ Running(H(o))  \\
  \forall o' \in dom(H).~o \notin waiters(H(o')) \land (o \in subscribers(H(o')) \Rightarrow o' \in \bar{p})
}
{ H \vdash F^{a(o, \bar{p})}~\textbf{ok}
}

\vspace{0.3cm}

\infrule[\textsc{H-ok}]
{ \forall o \in dom(H).~H \vdash H(o) \ok \\
  \forall o_1 \neq o_2 \in dom(H).~waiters(H(o_1))\#waiters(H(o_2))
}
{ \vdash H \ok
}

\vspace{0.3cm}

\infrule[\textsc{HO-ok}]
{
}
{ H \vdash \obj C {FM} \ok
}

\vspace{0.3cm}

\infrule[\textsc{DOHO-ok}]
{
}
{ H \vdash \langle \texttt{Observable$[\sigma]$}, done(\bar{S}) \rangle \ok
}

\vspace{0.3cm}

\infrule[\textsc{ROHO-ok}]
{ \forall i \neq j \in \{1..n\}.~obsIds(F_i)\#obsIds(F_j)  \\
  \forall i \in \{1..n\}.~\forall o \in obsIds(F_i).~Running(H(o))
}
{ H \vdash \langle \texttt{Observable$[\sigma]$}, running(F_1, \ldots, F_n, \bar{S}) \rangle \ok
}

\vspace{0.3cm}

\infrule[\textsc{Proc-ok}]
{ H \vdash FS_1 \ok \quad \ldots \quad H \vdash FS_n \ok  \\
  \forall i \neq j \in \{1..n\}.~obsIds(FS_i)\#obsIds(FS_j)
}
{ H \vdash \{ FS_1, \ldots, FS_n \} \ok
}

  \caption{Non-interference properties.}
  \label{fig:non-interference}
\end{figure*}

Rule \textsc{FS-ok} requires that the observable ids in a frame stack are
distinct. Rule \textsc{AF-ok} requires that an observable id $o$ in the label
of an asynchronous frame is not included in the observable ids of the waiters
of any observable in the heap; moreover, if $o$ is included in the observable
ids of the subscribers of another observable $o'$, then $o'$ must be included
in the set of subscriptions of $o$, $\bar{p}$. Rule \textsc{H-ok} requires
that for all observable objects in the heap, the observable ids of the waiters
are disjoint. Rule \textsc{ROHO-ok} requires that a running observable object
(a) has no duplicate observable ids in its waiters, and (b) has only waiters
that refer to running (non-terminated) observables. Finally, rule \textsc{Proc-ok}
requires that the observable ids in the frame stacks of a process are pairwise
disjoint.

To enforce non-interference during evaluation we define a relation between
heaps. The following relation also (a) preserves the types of heap objects and
(b) bounds the observable ids of new running states.

\begin{definition}[Heap Evolution]
Heap $H$ evolves to $H'$ wrt a set of observable ids $B$, written $H \leq_B H'$ if

\begin{enumerate}

\item[(i)] $\forall o \in dom(H').$ if $o \notin dom(H)$ and $H'(o) = \langle \psi, running(\bar{F}, \bar{S}) \rangle$ then $\bar{F} = \bar{S} = \epsilon$, and

\item[(ii)] $\forall o \in dom(H).$
  \begin{itemize}
  \item if $H(o) = \obj C {FM}$ then $H'(o) = \obj C {FM'}$,
  \item if $H(o) = \langle \psi, done(\bar{S}) \rangle$ then $H'(o) = \langle \psi, done(\bar{R} \uplus \{ \langle o', q' \rangle \}) \rangle$ where $\bar{S} = \bar{R} \uplus \{ \langle o', q \rangle \}$, and
  \item if $H(o) = \langle \psi, running(\bar{F}, \bar{S}) \rangle$ then $H'(o) = \langle \psi, running(\bar{F}, \bar{S} \uplus \{ \langle o, [] \rangle \}) \rangle$ or ($H'(o) = \langle \psi, running(\epsilon, \bar{R}) \rangle$ and $dom(\bar{S}) = dom(\bar{R})$) or ($H'(o) = \langle \psi, running(\bar{F} \cup \bar{G}, \bar{S}) \rangle$, \newline $obsIds(\bar{F}) \# obsIds(\bar{G})$ and $obsIds(\bar{G}) \subseteq B$) or $H'(o) = \langle \psi, done(\bar{S}) \rangle$.
  \end{itemize}
\end{enumerate}

\end{definition}

\subsection{Subject Reduction}

The following subject reduction theorem is based on a typing relation that is
given for processes, frame stacks, frames, expressions, and heaps.
Figure~\ref{fig:type-rules-expressions} shows the typing relation for
expressions. It is straight-forward to extend this relation to frames, frame
stacks, and processes.

\begin{theorem}[Subject Reduction]
If $\vdash H \typ \star$ and $\vdash H \ok$ then:
\begin{enumerate}

\item If $H \vdash F \typ \sigma$, $H \vdash F \ok$ and $\freduce H F {H'} {F'}$ then
      $\vdash H' \typ \star$, $\vdash H' \ok$, $H' \vdash F' \typ \sigma$, $H' \vdash F' \ok$,
      and $\forall B.~H \le_B H'$.

\item If $H \vdash FS \typ \sigma$, $H \vdash FS \ok$ and $\fsreduce H {FS} {H'} {FS'}$ then
      $\vdash H' \typ \star$, $\vdash H' \ok$, $H' \vdash FS' \typ \sigma$, $H' \vdash FS' \ok$ and
      $H \le_{obsIds(FS)} H'$.

\item If $H \vdash P \typ \star$, $H \vdash P \ok$ and $\reduce H P {H'} {P'}$ then
      $\vdash H' \typ \star$, $\vdash H' \ok$, $H' \vdash P' \typ \star$ and $H' \vdash P' \ok$.
\end{enumerate}
\end{theorem}

\begin{proof}

Part (1) is proved by induction on the derivation of $\freduce H F {H'} {F'}$.
Part (2) is proved by induction on the derivation of $\fsreduce H {FS} {H'}
{FS'}$ and part (1). Part (3) is proved by induction on the derivation of
$\reduce H P {H'} {P'}$ and part (2).

\end{proof}

% \subsection{Typing relation}

\begin{figure*}[ht!]
  \centering

\infrule[\textsc{T-Boolean}]
{
}
{ \Gamma ; \Delta \vdash \uline{b} \typ \texttt{Boolean}
}

\vspace{0.3cm}

\infrule[\textsc{T-Int}]
{
}
{ \Gamma ; \Delta \vdash \uline{i} \typ \texttt{Int}
}

\vspace{0.3cm}

\infrule[\textsc{T-Var}]
{ x : \rho \in \Gamma
}
{ \Gamma ; \Delta \vdash x \typ \rho
}

\vspace{0.3cm}

\infrule[\textsc{T-Null}]
{
}
{ \Gamma ; \Delta \vdash \texttt{null} \typ \rho
}

\vspace{0.3cm}

\infrule[\textsc{T-Cond}]
{ \Gamma ; \Delta \vdash x \typ \texttt{Boolean} \andalso \Gamma ; \Delta \vdash e \typ \rho \andalso \Gamma ; \Delta \vdash e' \typ \rho
}
{ \Gamma ; \Delta \vdash \texttt{if}~(x)~\{e\}~\texttt{else}~\{e'\} \typ \rho
}

\vspace{0.3cm}

\infrule[\textsc{T-Field}]
{ ftype(\sigma, f) = \tau
}
{ \Gamma , x : \sigma ; \Delta \vdash x.f \typ \tau
}

\vspace{0.3cm}

\infrule[\textsc{T-New}]
{ fields(C) = \bar{f} \andalso ftype(C, \bar{f}) = \bar{\tau} \andalso \Gamma ; \Delta \vdash \bar{y} : \bar{\tau}
}
{ \Gamma ; \Delta \vdash \texttt{new}~C(\bar{y}) \typ C
}

\vspace{0.3cm}

\infrule[\textsc{T-Invoke}]
{ mtype(\sigma, m) = (\bar{\tau}) \rightarrow \sigma' \andalso \Gamma , x : \sigma \vdash \bar{y} : \bar{\tau}
}
{ \Gamma , x : \sigma ; \Delta \vdash x.m(\bar{y}) \typ \sigma'
}

\vspace{0.3cm}

\infrule[\textsc{T-RAsync}]
{ \Gamma ; \Delta \vdash \bar{y} : \texttt{Observable[}\bar{\tau}\texttt{]} \andalso \Gamma ; \sigma \circ \Delta \vdash e \typ \rho
}
{ \Gamma ; \Delta \vdash \texttt{rasync[}\sigma\texttt{]}(\bar{y})~\{e\} \typ \texttt{Observable[}\sigma\texttt{]}
}

\vspace{0.3cm}

\infrule[\textsc{T-Await}]
{ \Gamma ; \Delta \vdash x \typ \texttt{Observable[}\sigma\texttt{]}
}
{ \Gamma ; \Delta \vdash \texttt{await}(x) \typ \sigma
}

\vspace{0.3cm}

\infrule[\textsc{T-Let}]
{ \Gamma ; \Delta \vdash e \typ \rho \andalso \Gamma , x \typ \rho ; \Delta \vdash t \typ \sigma
}
{ \Gamma ; \Delta \vdash \texttt{let}~x = e~\texttt{in}~t \typ \sigma
}

\vspace{0.3cm}

\infrule[\textsc{T-While}]
{ \Gamma ; \Delta \vdash x \typ \texttt{Boolean} \andalso \Gamma ; \Delta \vdash t \typ \rho
}
{ \Gamma ; \Delta \vdash \texttt{while}~(x)~\{t\} \typ \texttt{Boolean}
}

\vspace{0.3cm}

\infrule[\textsc{T-Yield}]
{ \Gamma ; \sigma \circ \Delta \vdash x \typ \sigma
}
{ \Gamma ; \sigma \circ \Delta \vdash \texttt{yield}(x) \typ \sigma
}

  \caption{Type checking expressions.}
  \label{fig:type-rules-expressions}
\end{figure*}

\bibliographystyle{plain}
\bibliography{bib}

\end{document}